\documentclass[fleqn,10pt]{wlscirep}
\usepackage{subfigure}
\usepackage{graphicx}
\usepackage{amsthm}
\usepackage{amsmath}
\usepackage{amssymb}
\usepackage{amsthm}

\newtheorem{mydef}{Proposition}

\begin{document}

\date{\today}

\title{A new method to reduce the number of time delays in a network}

\author[1,*]{Alexandre~Wagemakers}
\author[1,2,3]{Miguel~A.F.~Sanju\'an}

\affil[1]{Nonlinear Dynamics, Chaos and Complex Systems Group, Departamento de  F\'isica, Universidad Rey Juan Carlos, M\'ostoles, Madrid, Tulip\'an s/n, 28933, Spain}

\affil[2]{Department of Applied Informatics, Kaunas University of Technology, Studentu 50-415, Kaunas LT-51368, Lithuania}

\affil[3]{Institute for Physical Science and Technology, University of Maryland, College Park, Maryland 20742, USA}

\affil[*]{alexandre.wagemakers@urjc.es}

\begin{abstract}

Time delays may cause dramatic changes to the dynamics of interacting oscillators. Coupled networks of interacting dynamical systems can have unexpected behaviours when the signal between the vertices are time delayed. It has been shown for a very general class of systems that the time delays can be rearranged as long as the total time delay over the constitutive loops of the network is conserved. This fact allows to reduce the number of time delays of the problem without loss of information. There is a theoretical lower bound for this number that can be numerically improved if the time delays are commensurable. Here we propose a formulation of the problem and a numerical method to even further reduce the number of time delays in a network.
\end{abstract}

\maketitle
\section{Introduction}

Transmission delays are intrinsic to any process that exchanges information. While in many applications these time delays are small enough to be neglected, in other cases they have a critical influence on the dynamics. Examples of connected dynamical systems appear frequently in physics, engineering and natural sciences \cite{soriano2013complex,liang09,perez11}. A problem of interest in the dynamical systems community is the synchronization of coupled oscillators. Some progress has been made to understand the synchronization of oscillators when an identical time delay is present on every connection of a coupled system \cite{yeung1999time,earl2003synchronization}. Otherwise, the problem of synchronization with nonidentical time delays spreaded accross a network is still difficult to address, yet there are succesfull attempts of analysis using a mean field approach of the dynamical system \cite{ott_delay,petkoski2016heterogeneity}.

 In an effort to simplify the study of such networks, a new method called componentwise time-shift transformation \cite{lucken_reduction_2013} has been developed in order to transform the time delays of the network. This transformation allows to change the time delays on the network following some precise rules without affecting the dynamics of the system \cite{lucken_reduction_2013,lucken_classification_2015}. The purpose of the transformation is to set $n-1$ time delays to zero, being $n$ the number of vertices of the network. A brief summary is described in the next section.

Here we take on this idea and propose a new formulation of this transformation that allows to use common optimization algorithms to reduce the number of time delays on a network. Our results show that on networks with different topologies, the number of time delays that can be reduced to zero is larger that $n-1$. We claim that the number $n_z$ of zero time delays can be larger than the lower bound $n_z = n-1$ in the case of a set of commensurable time delays. Moreover, within our framework we can devise other optimization strategies to find a suitable configuration of time delays given a specific need.

The technique described in \cite{lucken_classification_2015} hinges on the observation that we can change the time delays in a network with a single cycle of length $n$ without altering the dynamics as long as the sum of the time delays around the cycle is conserved. The authors extended the reasoning over arbitrary networks by establishing the constitutive constraints between the time delays in a network. In this article, we reformulate the fundamental property of conservation of the time delays over a loop using algebraic graph theory. The problem of finding minimal time delays on the network is next transformed into a linear optimization problem. We show that the simplex optimization algorithm \cite{bazaraa2011linear} provides a solution for the transformed time delays where at least $n-1$ time delays are set to zero.

\section{Componentwise time-shift transformation}
We consider a graph $G$ with a collection of $l$ oriented edges $e_{i}$ and $n$ vertices $v_i$. At each vertex we have a very general dynamical system in the form of a system of $n$ delay differential equations
\begin{equation} \label{sys_din}
  \frac{dx_i}{dt} = f_i(x_i, x_j(t-\tau_{k})_{k \in S_i}),
\end{equation}
with $i=1,...,n$ and $S_i$ is the set of indices $k$ such that the edges $e_k$ connect the vertex $j$ to the vertex $i$. We assume a discrete time delay $\tau_{k}$ on the edge $e_k$.

The previous system in Eq.~(\ref{sys_din}) can be transformed with a redefinition of the time delays $\tau_{k}$ without changing the dynamical properties of the system. We set
\begin{equation} \label{sys_din_shift}
\frac{dy_i}{dt} = f_i(y_i, y_j(t-\tilde \tau_{k})_{k \in S_i}),
\end{equation}
with the following change of variables
\begin{align} \label{delay_shift}
y_i(t) = x_i(t-\eta_i) \\
\tilde \tau_{k} = \tau_{k} + \eta_{s(k)} - \eta_{t(k)},
\end{align}
being $\eta_i$ constants and $s(k)$ is the source vertex of the edge $k$ and $t(k)$ the target vertex of the same edge. The authors in~\cite{lucken_classification_2015} noticed that the {\it algebraic sum} of the time delay around any cycle of the network is constant for every choice of the time-shifts $\eta_i$. The term {\it algebraic sum} means here that, given an oriented cycle in the graph, the time delay associated to the edges on the cycle with the same orientation should be summed up and the time delays on edges with opposite direction subtracted.

Now the problem is to find the time-shifts $\eta_i$ associated to each vertex for a desired configuration of time delays $\tilde \tau_{k}$.

\section{Graph characteristics}

The topology of the graph can be described in terms of algebraic structures associated to the topology \cite{biggs1993algebraic}. We first give some definitions to set the context of the work. We define $G(V,E,A)$ as a directed and connected graph, where $V$ is a set of $n$ vertices and $E$ a set of $l$ directed edges. Before going into the details, we need to number the edges from 1 to $l$ and we note $\tau_k$ as the time delay of the edge $e_k$.

To represent the connectivity, we define the incidence matrix $A \in \mathbb{Z}^{n\times l}$  that relates the vertices to the edges. The elements $a_{jk}$ of the matrix $A$ are expressed in the following way: $a_{jk}=1$ if the edge $e_{k}$ points towards the vertex $j$ and $a_{jk}=-1$ if the edge points outwards. All other entries are zero. All the information about the connections of the network is contained in this matrix. It is also possible to develop the method for multiple edges connecting two vertices. We restrain here the case to a maximum of two edges to represent a bidirectional connection.

If the graph is connected, or weakly connected, we can define an acyclic subgraph called spanning tree that connects all the vertices and have exactly $n-1$ edges. This structure is important for the decomposition of the graph $G$ into elementary cycles. Given a spanning tree $T$ and an edge $e$ not in $T$, there is a unique cycle in $G$ containing only edges of $T$ and $e$. As a consequence, we can decompose the network into $c=l-(n-1)$ independent cycles. This decomposition can be expressed as a matrix $B \in \mathbb{Z}^{(l-(n-1))\times l}$ that expresses the cycle space associated to the tree $T$. First, we set the orientation of the cycle as the direction of the edge {\emph not} in $T$. Being $b_{jk}$ an element of $B$, we set $b_{jk}=1$ if the edge $k$ is in the cycle $j$ with the same direction, and $b_{jk}=-1$ if the orientations are opposite. All other numbers are zero. This matrix $B$ is of special interest for our study since the sum of time delays around each cycle is given by a simple matrix multiplication
\begin{equation}
B \boldsymbol \tau ={\boldsymbol \sigma},
\end{equation}
where $\boldsymbol \tau =(\tau_1 \dots \tau_l)^\intercal$ is the column vector of the time delay $k$ associated to the edge $e_k$ and $v^\intercal$ denotes the transpose of the vector $v$. The vector $\boldsymbol \sigma$ is what matters for the dynamics of the coupled system of delay differential equations. The time delays can be shuffled and changed into a new vector $\tilde{\boldsymbol \tau}$, but the vector $\boldsymbol \sigma$ should be constant, so that
\begin{equation} \label{eq:lineal_del}
B \boldsymbol \tau =B \tilde{\boldsymbol \tau}.
\end{equation}
This is the key property of the graph that we need to explore the space of possible solutions of $\tilde {\boldsymbol \tau}$.

The last necessary step to obtain the full characterization of the system is to derive the time-shifts $\eta_i$ in Eq.~(\ref{delay_shift}) that can lead us back to the time series of the original configuration of Eqs.~(\ref{sys_din}). These time-shifts $\eta_j$ associated to the vertex $j$ in the network is computed from a recursive relation on the spanning tree $T$ \cite{lucken_classification_2015},
\begin{equation} \label{eq_rec}
\tilde \tau_k - \tau_k = \eta_{s(k)} - \eta_{t(k)},
\end{equation}
with $s(k)$ the source vertex of the edge $k$ and $t(k)$ the target vertex of the same edge. Since the time-shifts are defined up to a constant, we can choose the value $\eta_1=0$ as a reference for all the other vertices. The rest of the time-shifts can be obtained from the incidence matrix $A$ restricted to the edges of the spanning tree. The columns of the incidence matrix contain exactly the values $s(k)$ and $t(k)$ for any edge $k$. If we partition the edges into two subsets of edges in and out of the spanning tree, we can rearrange the incidence matrix in two blocks:
\begin{equation} \label{split_incidence}
A= \left[\begin{array}{c|c} A_{in} & A_{out} \end{array}\right].
\end{equation}
The first matrix $A_{in} \in \mathbb{Z}^{n \times (n-1)}$ contains the information on the edges in the spanning tree. We also split the time delays in two similar sets ${\boldsymbol \tau}_{in}$ and ${\boldsymbol \tau}_{out}$ and we define a column vector with the time-shifts $\boldsymbol \eta=(\eta_1 \dots  \eta_n)^\intercal$. From the recursive relation given in Eq.~(\ref{eq_rec}), we can infer that
\begin{equation}
{\boldsymbol {\tilde \tau}}_{in} -{\boldsymbol \tau}_{in} = A_{in}^\intercal \boldsymbol \eta.
\end{equation}
However, we are looking for the time-shifts $\boldsymbol \eta$ as a function of the time delays of the tree $T$. Notice that the matrix $A_{in}$ has a rank $n-1$ and that the vector $\boldsymbol \eta$ has $n-1$ unknowns since $\eta_1=0$. Consequently, we can construct a full rank square matrix $A_{r}^\intercal$ by removing the first column of $A_{in}^\intercal$. We define the vector $\boldsymbol \eta^- = ( \eta_2 \dots \eta_n )^\intercal$ and transform the last equation into:
\begin{equation}
\boldsymbol{\tilde \tau}_{in} -\boldsymbol \tau_{in} = A_{r}^\intercal \boldsymbol \eta^- .
\end{equation}
Now we have a linear system with a single solution:
\begin{equation}
\boldsymbol \eta^-  = (A_{r}^\intercal)^{-1} (\boldsymbol{\tilde \tau}_{in} -\boldsymbol \tau_{in}) .
\end{equation}

The matrices $A$ and $B$ in Eq. (\ref{eq:lineal_del}) and (\ref{split_incidence}) are straightforward to derive and have a strong dependence to each other \cite{deo_graph_2016}. Notice also that the matrices $A_{in}$, $A_{out}$ and $B$ depend on the initial chosen spanning tree. We can demonstrate that this choice does not affect the space of possible solutions that can be reached with Eq.~(\ref{eq:lineal_del}). Any spanning tree can give us a valid basis to reconfigure the time delays in the network.

\section{Optimization of network time delays}

The main problem is stated in Eq.~(\ref{eq:lineal_del}) where all the possible vectors $\boldsymbol{\tilde \tau}$ are contained. We have to restrain however the problem to positive time delays $\tau_k$ to avoid complications with negative time delays. This consists of finding a vector $\boldsymbol{\tilde \tau}$ that will minimize the sum of the time delays over the network. This problem takes naturally the form of a standard linear program, that is,
\begin{equation} \label{lp_pb}
\begin{array}{rl}
\textrm{Minimize: }& {\textstyle \sum} \tilde \tau_k\\
&\\
\textrm{Constrained to: }& B \boldsymbol{\tilde \tau} = \boldsymbol{\sigma}\\
& \tilde \tau_k \geq 0
\end{array}
\end{equation}
This standard linear program can be solved with conventional techniques such as the simplex optimization algorithm \cite{bazaraa2011linear}. We show that the simplex method reduces the time delays of the network with at least $n-1$ zero time delays.

\begin{mydef}
Using the simplex algorithm, we can guarantee that there is a feasible solution $\boldsymbol{\tilde \tau}$ to the problem in  Eq. (\ref{lp_pb}) such that {\emph at least} $n-1$ time delays in the vector $\boldsymbol{\tilde \tau}$ are set to zero.
\end{mydef}

\begin{proof}
In the simplex algorithm, there is a first search for a basic feasible solution to the problem in a $l$-dimensional space. For such a solution, the columns of the matrix $B$ are rearranged into $[ D | Z ]$ where $D$ is an invertible $c\times c$ matrix and $Z$ is  a $c \times (n-1)$ matrix. The vector  $\boldsymbol{\tilde \tau} =(\boldsymbol \tau_D \boldsymbol \tau_Z)$ solution to the equation $B \boldsymbol{\tilde \tau} = \boldsymbol{\sigma}$ can be decomposed into $\boldsymbol \tau_D=D^{-1} \boldsymbol{\sigma}$ and $\boldsymbol \tau_Z= \boldsymbol 0$ a vector with all zeros. Being $n-1$ the size of the vector $\boldsymbol \tau_Z$, we have a valid reduction of the network with $n-1$ time delays set to zero. The other part $\boldsymbol \tau_D$ contains only positive time delays.
\end{proof}
The existence of one basic feasible solution gives us a valid reduction, however the algorithm looks further for an optimal solution minimizing the sum of the time delays. The solver will find a solution with $n_z\geq n-1$ and a total sum of the time delays below or equal than the initial sum of the time delays $\sum \tau_k$. There are plenty of efficient implementations of the simplex algorithm to solve linear programs \cite{bazaraa2011linear}, and we can obtain the reduction of the network in a polynomial time.

We now have all the ingredients to construct a optimized network. All we need is any spanning tree $T$, the incidence matrix $A$ and a fundamental loop matrix $B$ of the graph $G$.

\section{Numerical experiments}

The results of the optimization method may vary with the topology of the network and the statistical distribution of the time delays. Initially, we will focus on the effects of the topology of the network on the optimization output by setting identical time delays on every edge. To quantify the results of the algorithm, we have selected two representative parameters: the ratio $r_z=n_z/(n-1)$ that represents the ratio of the number of zero time delays over the lower bound $n-1$, and the ratio
\begin{equation} \label{def_rs}
r_s=1-\frac{\sum\limits_{k=1}^l \tilde \tau_k}{\sum\limits_{k=1}^l \tau_k},
\end{equation}
which measures the reduction of the sum of the time delays after the optimization process. A number $r_s=0$ means that the algorithm was unable to find a lower sum of time delays, while $r_s=1$ happens when all the time delays have been reduced to zero.

To understand the role of the topology we have centered our attention on two markers: the density $\rho$ of the graph and the second moment of the degree distribution $\langle k^2\rangle$ as a measure of degree heterogeneity \cite{small2014randomnet,newman2010networks}. The density of the graph is the ratio between the actual and the maximum number of edges possible for a given number of vertices. In average, it is equal to $\rho=\langle k\rangle/(n-1)$ where $\langle k\rangle$ is the mean vertex degree.

\begin{figure}
\begin{center}
\includegraphics[width=10cm]{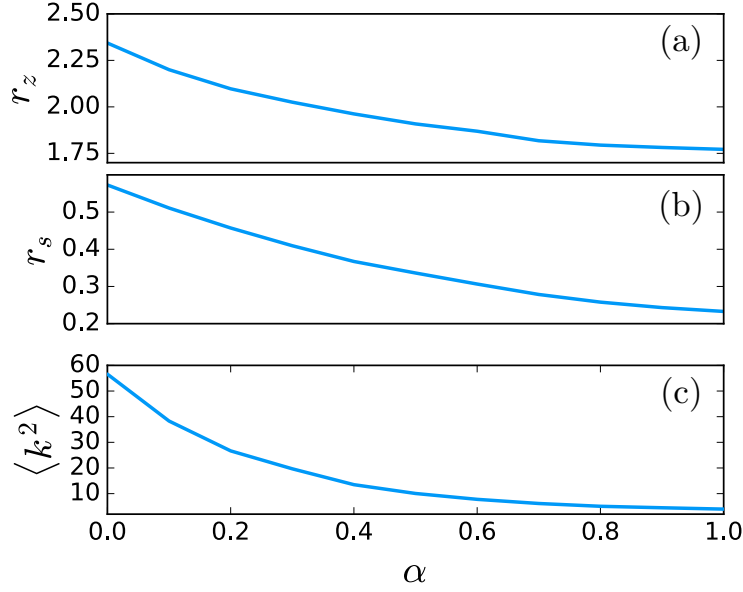}
\caption{\label{fig1} {\bf Influence of the second moment of the vertex degree distribution $\langle k^2\rangle$ on the optimization}. The variance $\langle k^2\rangle$ is modulated with a parameter $\alpha$. The plots represent in (a) the ratio $r_z=n_z/(n-1)$, in (b) the ratio $r_s$ defined in Eq. (\ref{def_rs}), and in (c) the variance of the vertex degree $\langle k^2\rangle$. The correlation between the variation of $\langle k^2\rangle$ and the two ratios $r_z$ and $r_s$ is clear. The simulations have been averaged over 30 networks of $n=400$ vertices with average degree $\langle k\rangle=4$ for each value of $\alpha$.}
\end{center}
\end{figure}

The influence of the degree heterogeneity of the network can be studied by interpolating a network between an Erd\"os-R\'enyi random network and a scale-free network \cite{gomez2006scale}. The shape of the degree distribution can be changed continuously between the two limiting cases with a single parameter $\alpha$. For $\alpha=0$ we obtain a scale-free network with an exponent $2.8$ of the power law degree distribution. When $\alpha=1$ we get a Erd\"os-R\'enyi random network with a probability $p=\langle k\rangle/n$ to have an edge between two vertices. The second moment of the degree distribution $\langle k^2\rangle$ diverges for the  scale-free network in the thermodynamic limit $n\to \infty$. For the Erd\"os-R\'enyi random network the degree distribution is binomial and the variance tends to $n p (1-p)$. So we can expect that $\langle k^2\rangle$ decreases when $\alpha$ takes values from 0 to 1. The results in Fig. \ref{fig1} show a clear correlation between the ratios $r_z$, $r_s$ and the variance $\langle k^2\rangle$ of the vertex degree. For this simulation, the network is sparse with a constant low density $\rho=0.01$. The average path length is almost constant for all the simulation and, as we will explain later, we have ruled out the influence of the clustering coefficient. It seems that the heterogeneity in the degree of the vertex has an important role in the possible outcome of the reduction process, although the exact relationship between the three measures $\langle k^2\rangle$,  $r_z$ and $r_s$ is elusive.

The density $\rho$ of the graph has also a significant influence on the measures $r_z$ and $r_s$. To illustrate this assertion we will build a network with an adjustable density. Nevertheless it would be good to eliminate the influence of the variance $\langle k^2\rangle$. Noticing that the Watts-Strogatz random graph model with a rewiring probability $p$ and mean degree $\langle k\rangle$ has a variance degree roughly equal to $\langle k^2\rangle \simeq \langle k\rangle p(1-p)\simeq \langle k\rangle p$ for a small $p$, we can construct a network with an arbitrary density and almost a constant degree variance $\langle k^2\rangle$. We start with a regular regular ring network where $n$ vertices are coupled to the $k$ nearest neighbors. In this case the network is symmetric and all the vertices have the same degree. We break this symmetry by rewiring each edge with a probability $p=\varepsilon/\langle k\rangle$ leading to a variance of the vertex degree approximately $\langle k^2\rangle \simeq \langle k \rangle p=\varepsilon$. This is very similar to the small-world model construction but we focus on tuning the density while keeping the variance of the vertex degree constant $\langle k^2\rangle$ and very low. We can also assure with this construction that the graph is weakly connected. The results summarized in Fig. \ref{fig2} clearly uncover the dependence between the density and the ratios $r_z$. On the one hand, $r_z$ is directly proportional to $\rho$ and on the other hand $r_s$ and $\rho$ seem to follow a power law functional as the log-log plot in Fig. \ref{fig2} (c) suggests. This depedence means that $\rho$, and consequently the mean degree $\langle k \rangle$, has relevant influence on the optimization.

\begin{figure}
\begin{center}
\includegraphics[width=11cm]{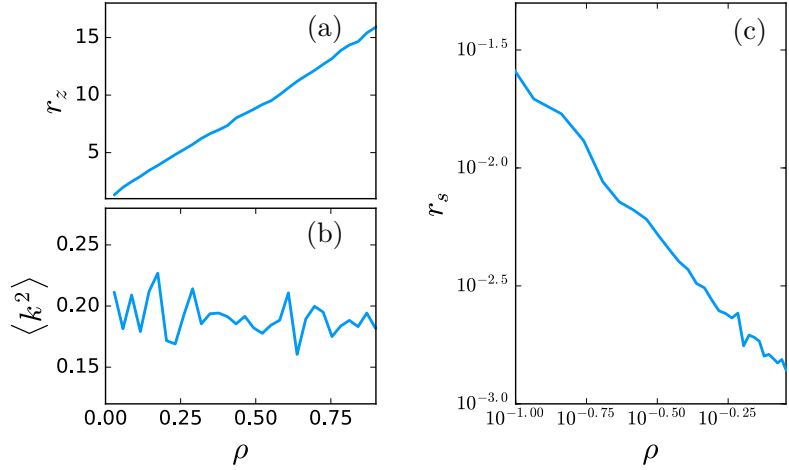}
\caption{\label{fig2} {\bf Influence of the density $\rho$ of the graph}. The figure shows the effect of the density $\rho=\langle k\rangle/(n-1)$ of a random graph on the optimized delays results. The panels represent in (a) the ratio $r_z=n_z/(n-1)$, in (b) the variance of the vertex degree $\langle k^2\rangle$, and (c) the ratio $r_s$. The ratio $r_z$ is directly proportional to $\rho$ and $r_s$ is inversely proportional to $\rho$. The simulations have been averaged over 40 realizations of a network with $70$ vertices for each value of $\rho$.}
\end{center}
\end{figure}

Other factors such as the clustering coefficient or the average path length do not seem to have a significant effect on the ratios $r_z$ and $r_s$. The clustering coefficient of a scale-free network has been tuned using a technique that adds triangles in the network without changing the degree distribution of the network \cite{holme2002growing}. While the clustering coefficient of the network evolves from 0 to 0.3, the ratios $r_z$ and $r_s$ remain almost unchanged. This is a counterintuitive result since it would have been reasonable to think that the presence of more triangles in the network would have brought an enhancement of the results.

\begin{figure}
\begin{center}
\includegraphics[width=10cm]{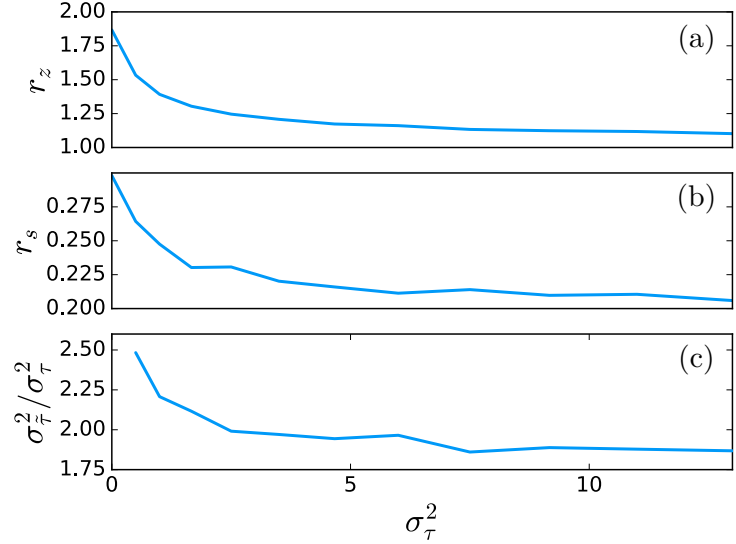}
\caption{\label{fig3} {\bf Influence of the variance $\sigma_\tau^2$ of delay distribution}. This plot shows the importance of the distribution of the time delay on the results of the algorithm. As the width of the distribution increases, the performance of the distribution get worse. The panels represent in (a) the ratio $r_z=n_z/(n-1)$, in (b) the ratio $r_s$, and (c) the ratio between the variance of the time delays $\sigma_{\tilde \tau}^2$ after and $\sigma_\tau^2$ before the optimization. The simulations have averaged over 60 realizations of a network with $n=400$ vertices with mean degree $\langle k\rangle=4$ for each value of $\sigma_\tau^2$.}
\end{center}
\end{figure}

The topology of the network has certainly a strong effect on the outcome of the optimization. However, the distribution of the time delays $\tau_k$ is also critical as shown in Fig. \ref{fig3} (a) and (b) where the measures $r_z$ and $r_s$ are represented as a function of the variance of the time delay distribution. The time delays have been chosen randomly in the discrete interval $[1; N_{max}]$ in order to control the variance $\sigma_\tau^2$ of the distribution. As this variance increases, the algorithm has more difficulty to find a solution with $r_z\geq 1$. The ratio $r_s$ also tends to diminish but it remains above 0. In Fig. \ref{fig3} (c) we can see that the ratio between the variance of the time delays $\sigma_{\tilde \tau}^2$ after and $\sigma_{\tau}^2$ before the optimization is larger than 1. In general the variance of the time delay distribution will increase after optimization but it seems that the ratio  $\sigma_{\tilde \tau}^2 /\sigma_{\tau}^2$  is bounded.

When the time delays are distributed following a continuous real valued distribution, it is almost impossible to find a solution with $n_z>n-1$. The simplex method finds only the basic feasible solution $n_z=n-1$, which is the minimum number of zero time delays achievable \cite{lucken_classification_2015}. If there are special relations between time delays, for example if they are all identical, it might be possible to reach a better solution. In the case of incommensurable real-valued time delays, such relations vanish. However, the simplex algorithm is still capable of finding a lower total sum of time delays, which may be of interest.

While real valued time delays are more general, integer valued time delays are very relevant when it comes to the numerical integration of differential equations. For the numerical algorithms involving finite and constant step size, the values of the time delays, that may have been issued from a continuous distribution, have to be discretized and rounded to the closest integer multiple of the time step. The set of continuous time delays is transformed into a new set o commensurable time delays that will give much better results from the point of view of the optimization.

The previous examples focus on the properties of the networks and delay distribution and do not involve any specific dynamical system. We present an application where a network of Kuramoto phase oscillators is coupled with time delays \cite{yeung1999time}. The phase oscillator model is a very simple abstraction of the essential properties of limit cycle oscillators. We can use this model to test our optimization method on a complex network of simple dynamical systems. The setup consists of a unidirectional Erd\"os-R\'enyi random network with average degree $d$, where the vertices represent Kuramoto oscillators with an identical intrinsic frequency $\omega$. The edges of the network represent a time delayed interaction chosen randomly according to a statistical distribution. The coupled delay differential equation can be written as
\begin{equation}
\frac{d\theta_i}{dt}=\omega + \frac{K}{d}\sum_{k\in S_i} (\theta_j(t-\tau_k) - \theta_i),
\end{equation}
where $S_i$ is the set of edges going from vertex $j$ to the vertex $i$ and $K$ is the coupling strength. We distribute the time delays $\tau_k$ following a uniform distribution in the continous interval $[\tau_m , 0.5 + \tau_m]$. Notice however that since we integrate the equation numerically, we have to discretize this interval, as said earlier, due to the finite time step size of the algorithm. In order to test the validity of the reduction in a dynamical system, we use the average frequency of the network since this measurement is independent of the initial history of the delay differential equation \cite{nordenfelt2014frequency}.

\begin{figure}
\begin{center}
\includegraphics[width=9cm]{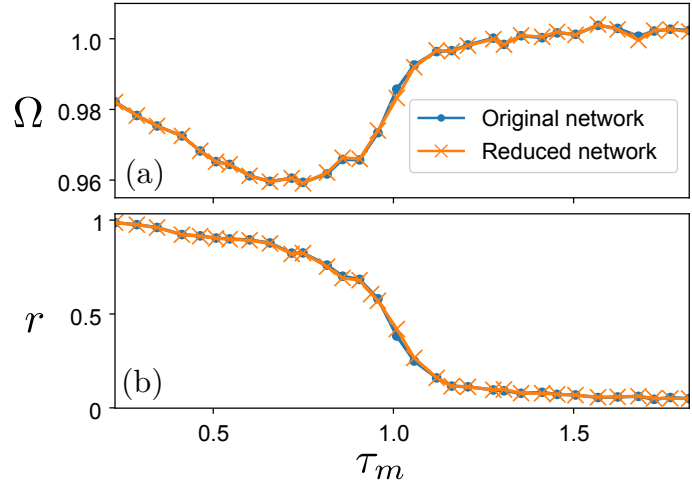}
\caption{\label{fig4} {\bf Average network network frequency $\Omega$  and order parameter $r$ of a coupled network of Kuramoto phase oscillators coupled with time delays.} The curves, that are superposed in Fig. (a), represent the average network frequency for the original (dot markers) and reduced network (cross markers). For each dot the average network frequency has been computed and averaged for several initial histories of the network to avoid numerical artifacts caused by the integration method. Both the original and reduced network lead to the same asymptotic frequency. In Fig. (b), the mean value of the order parameter has been represented for the two sets of simulations. Here again both curves agree on the same synchronization value for the two kind of network. Parameters are: $\omega = 1 $, $K=0.1$, $\tau_k \in [\tau_m; \tau_m+0.5]$, $n=50$, $d=4$.}
\end{center}
\end{figure}

We let evolve the network in time and we compute the average frequency $\Omega_i$ of each oscillator over a finite interval of time $T$
\begin{equation}
\Omega_i=\frac{1}{T}\int_0^T \dot \theta_i ~  dt.
\end{equation}
Then we compute the average network frequency $\Omega$ in this manner
\begin{equation}
\Omega=\frac{1}{n}\sum_i \Omega_i.
\end{equation}
This last frequency is independent of the chosen initial conditions and should be the same for both the original network and the reduced network given by Eq. (\ref{lp_pb}). In Fig. \ref{fig4} (a), we show an example where a network of $n=50$ oscillators has been simulated with a realization of the random time delays. The average frequency of the original and reduced network are consistent in both simulations showing that the asymptotic behavior is the same.

Another quantity of interest in the study of coupled oscillators is the synchronization order parameter
\begin{equation} \label{ord_param}
r(t) = \frac{1}{n} \left | \sum_{j=0}^n e^{i \theta_j(t)} \right|.
\end{equation}
This parameter can be averaged over time to characterize the state of the network with a single number
\begin{equation}
r = \frac{1}{T} \int_0^T r(t) dt.
\end{equation}
We cannot compare directly the order parameters of the original and reduced network since the time series are related through the change of variable in Eq. (\ref{delay_shift}). Being $\tilde \theta_j$ the variables of the reduced system, we can compute the order parameter introducing the time-shifts $\eta_j$ in the Eq. (\ref{ord_param})
\begin{equation} \label{ord_param}
\tilde r(t) = \frac{1}{n} \left | \sum_{j=0}^n e^{i \tilde\theta_j(t)} e^{-i \eta_j} \right|.
\end{equation}
Figure \ref{fig4} (b) represents the average order parameter $r$ and $\tilde r$ for the original and reduced network for the same parameters as the previous example. Both results overlap almost exactly meaning that the dynamics in the reduced system is conserved.

The simulations have been performed with the programming language Julia \cite{julialang} using LightGraphs, JuMP and Coin-or Linear Programming (Clp) packages.

\section{Conclusions}

Reorganizing the time delays in a network does not seem to be an easy task at first sight. But once the basic mechanisms of time delay conservation are understood, it is possible to change the time delays and at the same time to conserve the dynamical properties of the network. Our formulation along with the componentwise time-shift transformation technique opens a way to reduce even further the time delay space. When the problem is stated in the form of a linear program, the simplex algorithm provides a higher number of zero time delays than the theoretical lower bound $n_z$, that corresponds to the dimension of the cycle space of the network. It also finds the solution with the lowest sum of time delays, which can represent a reduction up to 60\% of the initial sum of the time delays.

The numerical integration of coupled dynamical systems with the presence of different time delays among the network usually involves a high computational and storage cost. The memory usage can be reduced up to 30\% with the optimization of the delay of the network. Another possible application is to modify the fitness function of the optimization algorithm such that the time delays fit a desired distribution more suitable to the problem at glance.

\bibliographystyle{naturemag}

\section*{Acknowledgements}

This work was supported by the Spanish Ministry of Economy and Competitiveness
under Project No. FIS2013-40653-P and by the Spanish State Research Agency (AEI)
and the European Regional Development Fund (FEDER) under
Project No. FIS2016-76883-P. MAFS acknowledges the
jointly sponsored financial support by the Fulbright Program
and the Spanish Ministry of Education (Program No. FMECD-ST-2016).

\section*{Author contributions statement}

A.W. and M.A.F.S. devised the research. A.W. performed the numerical simulations. A.W., and M.A.F.S. analyzed the results and wrote the paper.

\section*{Additional information}

\textbf{Competing financial interests}: The authors declare no competing financial interests.

\end{document}